\documentclass[11pt]{article}

\usepackage{graphicx,amssymb,amsmath}
\usepackage{epsfig}
\usepackage{epsf}
\usepackage[small]{caption}
\usepackage{amsfonts}
\usepackage{latexsym}
\usepackage{fancyvrb}
\usepackage{amssymb}
\usepackage{amsthm}
\usepackage{graphics}
\usepackage{psfrag}
\usepackage{float}
\usepackage{bbm}

\newtheorem{theorem}{Theorem}
\newtheorem{lemma}{Lemma}

\newtheorem{observation}{Observation}

\newcommand{\old}[1]{{}}
\newcommand{\later}[1]{{}}

\def\ie{{i.e.}}

\newcommand{\eps}{\varepsilon}

\def\L{\mathcal L}
\def\R{\mathcal R}

\newcommand{\len}{{\rm len}}
\newcommand{\per}{{\rm per}}
\newcommand{\conv}{{\rm conv}}

\newcommand{\longest}{{\rm long}}

\title{The traveling salesman problem for lines and rays in 
the plane\footnote{An earlier version of this paper
appeared in the {\em Proceedings of the 22nd Canadian Conference on
Computational Geometry} (CCCG 2010), Winnipeg, Manitoba, Canada,
August 2010, pp.~257--260.}}

\author{
Adrian Dumitrescu\thanks{Department of Computer Science,
University of Wisconsin--Milwaukee, WI 53201-0784, USA\@.
Email:~\texttt{dumitres@uwm.edu}.
Supported in part by NSF CAREER grant CCF-0444188
and by NSF grant DMS-1001667.
}}

\begin{document}


\maketitle

\begin{abstract}
In the Euclidean TSP with neighborhoods (TSPN), we are given a
collection of $n$ regions (neighborhoods) and we seek a shortest
tour that visits each region. In the path variant, we seek
a shortest path that visits each region. We present several
linear-time approximation algorithms with improved ratios for these
problems for two cases of neighborhoods that are (infinite) lines, 
and respectively, (half-infinite) rays.
Along the way we derive a tight bound on the minimum perimeter of 
a rectangle enclosing an open curve of length $L$. 

\medskip
\noindent\textbf{\small Keywords}:
Traveling salesman problem with neighborhoods,
linear programming,
minimum-perimeter rectangle,
approximation algorithm,
lines,
rays.

\end{abstract}

\section{Introduction}  \label{sec:intro}

In the Euclidean Traveling Salesman Problem (TSP),
given a set of points in the plane, one seeks a shortest
{\em tour} (closed curve) that visits each point. 
In the path variant, one seeks a shortest {\em path} (open curve) that
visits each point. If now each point is replaced by a 
(possibly disconnected) region, one obtains the so-called 
TSP with {\em neighborhoods} (TSPN), first studied by Arkin and
Hassin~\cite{AH94}. A tour for a set of neighborhoods is also
referred to as a TSP tour. A path for a set of neighborhoods is also
referred to as a TSP path.

For the case of neighborhoods that are (infinite) straight lines, an
optimal tour can be computed in $O(n^5)$ time~\cite{CJN99,THI99,T01} 
(see also~\cite{J02}), and a $\sqrt{2}$-approximation can be computed
in $O(n)$ time~\cite{J02}.  
For the case of neighborhoods that are (half-infinite) rays, 
no polynomial time algorithm is known for computing an optimal tour,
and a $\sqrt{2}$-approximation can be computed in $O(n)$ time~\cite{J02}. 
In this paper we present linear-time approximation algorithms 
with improved ratios for these problems. The obvious motivation 
is to provide faster and conceptually simpler algorithmic
solutions. As mentioned above, while for the case of rays 
no polynomial time algorithm is known, for the case of lines, 
the known algorithms reduce the problem of computing an optimal
tour of the lines to that of computing an optimal watchman tour in a
simple polygon for which the existent algorithms are quite
involved and rather slow for large $n$~\cite{CJN99,THI99,T01}.  

In this paper we present four improved linear-time approximation
algorithms for TSP, for two cases of neighborhoods, that are straight
lines, and respectively, straight rays in the plane. We consider 
two variants of the problem: that of computing a shortest tour and
that of computing a shortest path visiting the input set. 
Our algorithms are all based on solving low-dimensional
linear programs. Our results are summarized in Table~\ref{table}.

\begin{theorem}\label{T1}
Given a set of $n$ lines in the plane:
{\rm (i)} A TSP tour that is at most $1.28$ times longer than the optimal can be
computed in $O(n)$ time. 
{\rm (ii)} A TSP path that is at most $1.42$ times longer than the optimal can be
computed in $O(n)$ time. 
\end{theorem}

For lines, the previous best approximations obtained in linear time were 
$\sqrt2 \approx 1.41$ and $2\sqrt2 \approx 2.82$, respectively~\cite{J02}.

\begin{theorem}\label{T2}
Given a set of $n$ rays in the plane:
{\rm (i)} A TSP tour that is at most $1.28$ times longer than
the optimal can be computed in $O(n)$ time. 
{\rm (ii)} A TSP path that is at most $2.24$ times longer than
the optimal can be computed in $O(n)$ time. 
\end{theorem}

For rays, the previous best approximation for tours was 
$\sqrt2 \approx 1.41$~\cite{J02} (obtained also in 
linear time, however this was the only approximation known), while for
paths there was no approximation known.

\medskip
\begin{table*}[t]
\begin{center}
\begin{tabular}{||c|c|c|c|c||} \hline
Ratio & Tour (old ratio) & Tour (new ratio) & Path (old ratio) &
Path (new ratio)  \\ \hline \hline
 Lines & $\sqrt2=1.41\ldots$ & $1.28$ &
$2\sqrt2=2.82\ldots$ & $1.42$ \\ \hline
 Rays & $\sqrt2=1.41\ldots$ & $1.28$ &
 $-$  & $2.24$ \\ \hline
\end{tabular}
\end{center}
\caption{Old and new approximation ratios. No approximation for paths
  on rays was reported in~\cite{J02}.}
\label{table}
\end{table*}

\paragraph{Preliminaries.}
We use the following terms and notations.
We denote by $x(p)$ and $y(p)$ the $x$ and $y$-coordinates 
of a point $p$. 
We say that point $q$ dominates point $p$ if $x(p) \leq x(q)$ 
and  $y(p) \leq y(q)$.
For a segment $s$, $\Delta{x}(s)$ and $\Delta{y}(s)$ denote 
the lengths of its horizontal and vertical projections. 
The convex hull of a planar set $A$ is denoted by $\conv(A)$.
The Euclidean length of a curve $\gamma$ is denoted by $\len(\gamma)$.
For a polygon $P$, let $\per(P)$ denote its perimeter. 
For a rectangle $Q$, let $\longest(Q)$ denote the length of a longest
side of $Q$. 
For a ray $\rho$, let $\ell(\rho)$ denote its supporting line. 

The inputs to the two variants of TSP we consider 
are a set of lines or a set of rays.  
Let $\L$ be a given set of $n$ lines, and 
let $T^*(\L)$ be an optimal tour (circuit) of the lines in $\L$. 
Let $\R$ be a given set of $n$ rays, and 
let $T^*(\R)$ be an optimal tour (circuit) of the rays in $\R$. 

Following the terminology from~\cite{DJ10,R95},
a polygon is an \emph{intersecting polygon} of a set of regions 
in the plane if every region in the set intersects the interior
or the boundary of the polygon. The problem of computing a
minimum-perimeter intersecting polygon (MPIP) for the case 
when the regions are line segments was first considered by 
Rappaport~\cite{R95} in 1995. As of now,
MPIP (for line segments) is not known to be polynomial, nor it is
known to be NP-hard. 

Since both lines and rays are infinite (\ie, unbounded regions) 
finding optimal tours $T^*(\L)$ and $T^*(\R)$ are equivalent to
finding minimum-perimeter intersecting polygons (MPIPs) 
for $\L$ and $\R$ respectively. 
We can assume without loss of generality that
not all lines in $\L$ are concurrent
at a common point (this can be easily checked in linear time),
thus $\per(T^*(\L)) > 0$. The same assumption can be made for the rays
in $\R$, thus $\per(T^*(\R)) > 0$. 

The following two facts are easy to prove;
see also~\cite{DJ10,DM03,R95}.

\begin{observation}\label{O1}
If $P_1$ is an intersecting polygon of $\L$,
and $P_1$ is contained in another polygon $P_2$,
then $P_2$ is also an intersecting polygon of $\L$. 
The same statement holds for $\R$. 
\end{observation}
\begin{observation}\label{O2}
$T^*(\L)$ is a convex polygon with at most $n$ vertices.
Similarly, $T^*(\R)$ is a convex polygon with at most $n$ vertices.
\end{observation}

A key fact in the analysis of the approximation algorithms for
computing tours is the following lemma. 
This inequality is implicit in~\cite{W93}; a more
direct proof can be found in~\cite{DJ10}.  

\begin{lemma} {\rm~\cite{DJ10,W93}.} \label{L1}
Let $P$ be a convex polygon. Then the minimum-perimeter rectangle $Q$
containing $P$ satisfies $\per(Q) \leq \frac{4}{\pi}\,\per(P)$.
\end{lemma}

\section{TSP for lines} \label{sec:lines}

In this section we prove Theorem~\ref{T1}.

\paragraph {TSP tours.}
We present a $\frac4{\pi}(1+\eps)$-approximation algorithm 
for computing a minimum-perimeter intersecting polygon
of a set $\L$ of $n$ lines, running in $O(n)$ time. 
If we set $\eps =1/200$, we get the approximation ratio $1.28$. 
For technical reasons (see below) we choose $\eps \in [1/300,1/200]$
uniformly at random, and the approximation ratio remains $1.28$. 
The algorithm combines ideas from~\cite{DJ10,DM03,J02}. 
As in ~\cite{DJ10}, we first use the fact (guaranteed by
Lemma~\ref{L1}) that every convex polygon $P$ is contained in some
rectangle $Q = Q(P)$ that satisfies $\per(Q) \le \frac{4}{\pi}\,\per(P)$.
In particular, this holds for the optimal tours $T^*(\L)$ and $T^*(\R)$. 
Then we use linear programming to compute a $(1+\eps)$-approximation
for the minimum-perimeter intersecting rectangle of $\L$ (as in 
\cite{DJ10}; see also~\cite{J02}).

\medskip
\noindent{\bf Algorithm~A1.}
\begin{itemize}
\item[]
Let $m = \lceil \frac{\pi}{4\eps} \rceil$.
For each direction $\alpha_i = i\cdot 2\eps$, $i=0,1,\ldots,m-1$,
compute a minimum-perimeter intersecting rectangle $Q_i$ of $\L$
with orientation $\alpha_i$.
Return the rectangle with the minimum perimeter over all $m$ directions.
\end{itemize}

We now show how to compute the rectangle $Q_i$ by linear programming.
By a suitable rotation (by angle $\alpha_i$) of the set $\L$ of lines
in each iteration $i \ge 1$, we can assume that the rectangle $Q_i$
is axis-parallel.  This can be obtained in $O(n)$ time (per iteration). 
Let $\{q_1,q_2,q_3,q_4\}$ be the four vertices of $Q_i$ 
in counterclockwise order, starting with the lower leftmost corner
as in Figure~\ref{f2}.
As in~\cite{J02}, let $\L=\L^- \cup \L^+$ be the partition of $\L$
into lines of negative slope and lines of positive slope. 
By setting $\eps \in [1/300,1/200]$ uniformly at random, 
in each iteration $i$, with probability $1$ there are no vertical
lines in (the rotated set) $\L$. 

Observe (as in~\cite{J02}), that a line in $\ell \in \L^+$ intersects
$Q_i$ if and only if $q_2$ and $q_4$ are separated by $\ell$ (points
on $\ell$ belong to both sides of $\ell$). Similarly, a line in $\ell
\in \L^-$ intersects $Q_i$ if and only if $q_1$ and $q_3$ are
separated by $\ell$. 
The objective of minimum perimeter is naturally expressed as a linear function.
The resulting linear program has $4$ variables $x_1,x_2,y_1,y_2$ 
for the rectangle $Q_i =[x_1,x_2] \times [y_1,y_2]$, and $2n+2$ constraints.

\begin{align*} \label{LP1}
\textup{minimize} \quad & 2(x_2-x_1) + 2(y_2 -y_1) \ \ \ \quad \quad \textup{(LP1)}
\\
\textup{subject to} \quad & \left\{
\begin{array}{lll}
y_2 \geq a x_1 +b, & \ell : y=ax+b \in \L^+ \\ 
y_1 \leq a x_2 +b, & \ell : y=ax+b \in \L^+ \\ 
y_1 \leq a x_1 +b, & \ell : y=ax+b \in \L^- \\ 
y_2 \geq a x_2 +b, & \ell : y=ax+b \in \L^- \\ 
x_1 \leq x_2  \\
y_1 \leq y_2 
\end{array}
\right.
\end{align*}

Let $Q^*$ be a minimum-perimeter intersecting rectangle of $\L$.
To account for the error made by discretization, we need the following
easy fact; see~\cite[Lemma 2]{DJ10}. 

\begin{lemma} {\rm~\cite{DJ10}.} \label{L2}
There exists an $i \in \{0,1,\ldots,m-1\}$ such that 
$\per(Q_i) \leq (1+\eps)\,\per(Q^*)$. 
\end{lemma}

By Observations~\ref{O1} and~\ref{O2}, and by Lemmas~\ref{L1} and~\ref{L2}, the
algorithm A1 computes a tour that is at most $1.28$ longer than the
optimal. The algorithm solves a constant number of $4$-dimensional
linear programs, each in $O(n)$ time~\cite{M84}.  The overall time is $O(n)$.

\paragraph {TSP paths.}
The key to the improvement is offered by the following.

\begin{observation}\label{O3}
Let $Q$ be a rectangle. Then $Q$ intersects a set of lines $\L$ if and only
if any three sides of $Q$ intersect $\L$. 
\end{observation}
\begin{proof}
Fix any three sides of $Q$: $\{s_1,s_2,s_3\}$ (each $s_i$ is a closed
segment). Now if $\ell$ is a line intersecting $Q$, then $\ell$ intersects 
at least two sides of $Q$, hence it intersects at least one element of 
$\{s_1,s_2,s_3\}$, as required. 
\end{proof}

The next lemma gives a quantitative upper bound on the total length of 
three shorter sides of a rectangle enclosing a curve. 

\begin{lemma} \label{L3}
Any open curve of length $L$ can be included in a rectangle $Q$,
so that $ \per(Q) - \longest(Q) \leq \sqrt{2} L$. This inequality is
the best possible.
\end{lemma}
\begin{proof}
Let $\gamma$ be an open curve of length $L=\len(\gamma)$, and let $a,b \in
\gamma$ be the two endpoints of $\gamma$. We can assume w.l.o.g.\ that
$ab$ is a horizontal segment, and let $Q=Q(\gamma)$ be a minimal axis-aligned
rectangle containing $\gamma$. Write $z=|ab|$. 
Let $w$ and $h$ the lengths of the horizontal and vertical sides of
$Q$, respectively (\ie, the width and
height of $Q$). It suffices to show that $w +2h \leq \sqrt{2} L$. 

By construction $\gamma$ meets each side of $Q$. 
We trace $\gamma$ from $a$ to $b$ and subdivide it into a finite
number of open sub-curves $\gamma_i$; the endpoints of each sub-curve 
$\gamma_i$ belong to two distinct sides of $Q$. We denote by $s_i$ the
segment connecting the two endpoints of $\gamma_i$. 
By concatenating these segments (in the same traversal order) we get 
a polygonal curve connecting $a$ and $b$. We call this (not
necessarily unique) curve, a polygonal curve induced by $\gamma$;
see Figure~\ref{f1}.
\begin{figure*}[hbtb]
\centerline{\epsfxsize=3.1in \epsffile{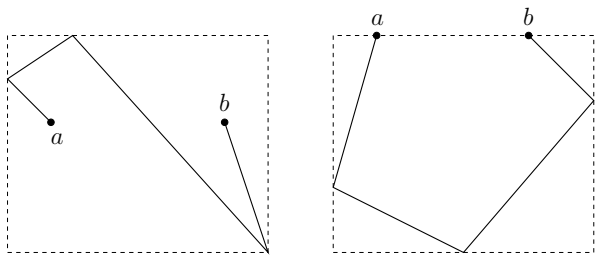}}
\caption{Polygonal curves  induced by $\gamma$ (two examples, in
  bold). The corresponding rectangles $Q$ are drawn in dashed 
  lines.} 
\label{f1}
\end{figure*}

For any segment $s$, we have $\Delta{x}(s) + \Delta{y}(s) \leq |s| \sqrt2$: 
indeed,  $\Delta{x}(s) + \Delta{y}(s) = |s|(\sin \alpha + \cos \alpha) \leq |s|
\sqrt2$, by a well-known trigonometric inequality
(here $\alpha \in [0,\pi/2]$). 
By adding the above inequalities for all segments $s_i$ (and sub-curves
$\gamma_i$) we obtain
\begin{equation} \label{E1}
\sum \left(\Delta{x}(s_i) + \Delta{y}(s_i)\right) 
\leq \left(\sum |s_i|\right) \sqrt2 
\leq \len(\gamma) \sqrt2 =L \sqrt2.
\end{equation}

On the other hand we have $\sum \Delta{x}(s_i) \geq w$: 
indeed, $\gamma$ starts at $a$ and meets the left and right sides of $Q$ 
before returning to $b$, hence the horizontal projections of the
segments $s_i$ sum up to at least $2w-|ab|=2w-z \geq w$.
Similarly, we have $\sum \Delta{y}(s_i) \geq 2h$: 
indeed, $\gamma$ starts at $a$ and meets the top and bottom sides of $Q$ 
before returning to $b$. 
Since $a$ and $b$ have the same $y$-coordinate, the vertical
projections of the segments $s_i$ cover twice the height of $Q$. 
Consequently, we have  
\begin{equation} \label{E2}
\sum (\Delta{x}(s_i) + \Delta{y}(s_i)) \geq w +2h.
\end{equation}

Putting~\eqref{E1} and~\eqref{E2} together yields
\begin{equation} \label{E3}
w + 2h \leq \sum \left(\Delta{x}(s_i) + \Delta{y}(s_i)\right) \leq L \sqrt2, 
\end{equation}
as required. 

To see that this inequality is tight, let $\gamma$ be a two-segment
polygonal path made from the two unit sides of an isosceles right
triangle. Then $L=\len(\gamma)=2$, while the rectangle $Q$ enclosing
$\gamma$ has sides $\sqrt2$ and $\sqrt2/2$ respectively. 
The lengths of its three shorter sides sum up to $\sqrt2
+ 2 \sqrt2/2= 2\sqrt2 =L \sqrt2$. It can be verified that the
sum of the three smallest sides of any other enclosing rectangle 
is larger (details in the Appendix), hence the rectangle $Q$
constructed in our proof is optimal for $\gamma$. The proof of
Lemma~\ref{L3} is now complete.  
\end{proof}

To compute a TSP path for a set of $n$ lines, we use the algorithm A2
we describe next. This algorithm is similar to algorithm A1, 
described earlier. A2 computes a rectangle in each direction from a
given sequence.   
The only difference in the linear program is that instead of
minimizing the perimeter of an intersecting rectangle, $ 2(x_2-x_1) +
2(y_2 -y_1) $, it minimizes the sum of the lengths of three 
sides, namely $ (x_2-x_1) + 2(y_2 -y_1) $. 
The objective function is not symmetric with respect to the two
coordinates axes, and so the number of directions $m$, from algorithm
A1, is $m = \lceil \frac{\pi}{2\eps} \rceil$ in algorithm A2.
Let now $Q^*$ be an intersecting rectangle of $\L$ with minimum sum of
the lengths of three sides. Analogous to Lemma~\ref{L2} we have

\begin{lemma} \label{L4}
There exists an $i \in \{0,1,\ldots,m-1\}$ such that 
$$ \per(Q_i) - \longest(Q_i) \leq 
(1+\eps)\,(\per(Q^*)- \longest(Q^*)). $$
\end{lemma}

By Lemma~\ref{L3} and Lemma~\ref{L4} the approximation ratio is
$\sqrt{2} (1+\eps)$, and we set $\eps =1/250$ (or slightly smaller,
as before), to obtain the approximation ratio $1.42$.  
This completes the proof of Theorem~\ref{T1}.

\section{TSP for rays} \label{sec:rays}

As noted in~\cite{J02}: If the lines are replaced by line segments
the problem of finding an optimal tour becomes NP-hard. Should the 
lines be replaced by rays, we get a variant of the problem that lies
somewhere in between the variant for lines and that for line segments,
and whose complexity is open. In this section we prove Theorem~\ref{T2}.

\paragraph {TSP tours.}
The algorithm A1 from Section~\ref{sec:lines} can be adapted to
compute a $\frac4{\pi}(1+\eps)$-approximate tour for a set $\R$ of $n$
rays. Let $m = \lceil \frac{\pi}{4\eps} \rceil$.
The resulting algorithm A3 for computing an approximate tour for 
$n$ given rays computes a minimum-perimeter rectangle intersecting 
all rays over all $m$ directions.
As before, assume that in the $i$th iteration, the rectangle 
$Q_i=\{q_1,q_2,q_3,q_4\}$ is axis-parallel. A ray in $\R$ is said to 
belong to the $i$th quadrant, $i=1,2,3,4$, if  its head belongs to the
$i$th quadrant when placed with its apex at the origin. 
Let $\R=\R_1 \cup \R_2 \cup \R_3 \cup \R_4$ be the partition of the
rays in $\R$ (after rotation) as dictated by the four quadrants. 
See Figure~\ref{f2}. 
\begin{figure}[htbp]
\centerline{\epsfxsize=2.3in \epsffile{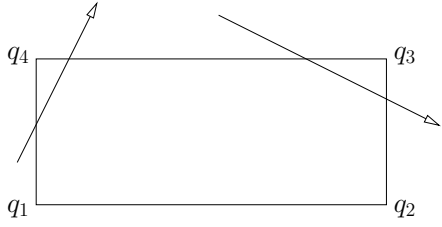}}
\caption{The rectangle $Q_i$, and two rays, one in $\R_1$ and
one in $\R_4$ that intersect it.}
\label{f2}
\end{figure}

Observe that:
\begin{itemize}
\item A ray $\rho \in \R_1$ intersects $Q_i$ if and only if $q_2$ and
$q_4$ are separated by $\ell(\rho)$, and the apex (endpoint) of $\rho$ is
dominated by $q_3$.  
\item A ray $\rho \in \R_2$  intersects $Q_i$ if and only if $q_1$ and $q_3$ are
separated by $\ell(\rho)$, and the apex of $\rho$ lies right and below $q_4$. 
\item A ray $\rho \in \R_3$ intersects $Q_i$ if and only if $q_2$ and $q_4$ are
separated by $\ell(\rho)$, and the apex of $\rho$ dominates $q_1$. 
\item A ray $\rho \in \R_4$  intersects $Q_i$ if and only if $q_1$ and $q_3$ are
separated by $\ell(\rho)$, and the apex of $\rho$ lies left and above $q_2$. 
\end{itemize}

The constraints listed above correct an error in the old
$\sqrt{2}$-approximation algorithm from~\cite{J02}, 
where it was incorrectly demanded 
that the apexes of the rays must lie in the rectangle $Q_i$. 
Indeed, this condition is not necessary, and moreover,
may prohibit finding an approximate solution with the claimed
guarantee of $\sqrt{2}$.

Observe that these intersection conditions can be expressed as linear
constraints in the four variables, $x_1,x_2,y_1,y_2$. 
The algorithm A3 computes a minimum-perimeter rectangle intersecting 
all rays over all $m$ directions. For each of these directions,
the algorithm solves a linear program with four variables 
and $O(n)$ constraints, as described above. As such, the algorithm
takes $O(n)$ time~\cite{M84}. The approximation ratio is
$\frac4{\pi}(1+\eps)$, and we set $\eps =1/200$ (or slightly smaller,
as before), to obtain the approximation ratio $1.28$.  

\paragraph {TSP paths.}
We need an analogue of Lemma~\ref{L2} for open curves. 
This is Lemma~\ref{L5} below which is obviously also of independent
interest. 

\begin{lemma} \label{L5}
Any open curve of length $L$ can be included in a rectangle $Q$,
so that $ \per(Q) \leq \sqrt{5} L$. This inequality is
the best possible.
\end{lemma}
\begin{proof}
Let $\gamma$ be an open curve of length $L=\len(\gamma)$, and let $a,b \in
\gamma$ be the two endpoints of $\gamma$. We can assume w.l.o.g.\ that
$ab$ is a horizontal segment, and let $Q=Q(\gamma)$ be a minimal axis-aligned
rectangle containing $\gamma$. Let $w$ and $h$ be the lengths of the 
horizontal and vertical sides of $Q$, respectively (\ie, the width and
height of $Q$). 

It suffices to show that $2w +2h \leq \sqrt{5} L$. 
Write $w=L/\lambda$, where $\lambda \geq 1$. 
The case $\lambda=1$ corresponds to a degenerate
enclosing rectangle, when $\gamma$ is a line segment,
and for which $2w +2h =2w = 2L \leq \sqrt{5} L$. 
We can therefore assume in the following that $\lambda>1$. 
By construction $\gamma$ meets each side of $Q$. 
Arbitrarily select a point of $\gamma$ on each of these sides
to obtain a polygonal open curve $\gamma_1$ connecting $a$ and $b$
and passing through these intermediate points (and still enclosed in $Q$). 
By construction, the intermediate points are visited in the same order by 
$\gamma$ and $\gamma_1$. 
By the triangle inequality, $\len(\gamma_1) \leq \len(\gamma)$. 

Successively reflect the rectangle $Q$ with respect to the sides
containing the intermediate points in the order of traversal. 
See Figure~\ref{f3} for an illustration. 
Let $b'$ be the final reflection of the end-point $b$ 
of $\gamma$. The segments composing $\gamma_1$ can be retrieved in the
reflected rectangles; they make up a polygonal path $\gamma_2$
connecting $a$ and $b'$. By construction we have 
$\len(\gamma_1) = \len(\gamma_2)$. 
\begin{figure}[htbp]
\centerline{\epsfxsize=5.2in \epsffile{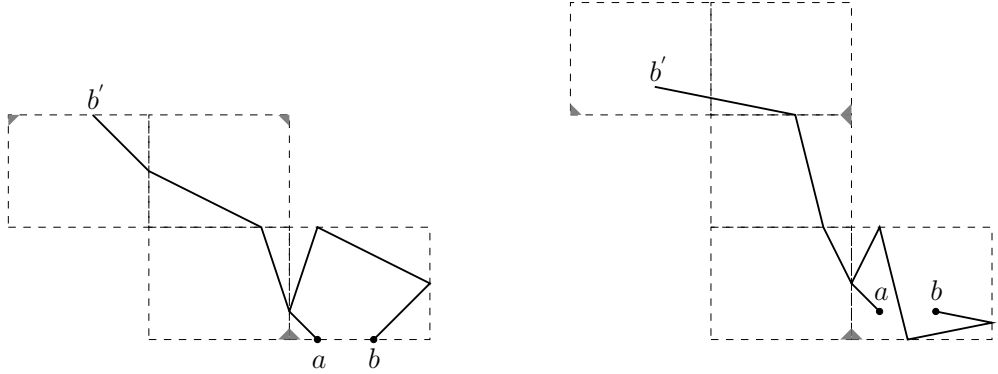}}
\caption{The reflection argument (two examples, in bold). The small
  shaded triangles indicate the lower left corner of $Q$ in all
  subsequent reflections.} 
\label{f3}
\end{figure}

It is easily seen that $\Delta{x} (ab') =2w -|ab|$ 
and $\Delta{x} (ab') =2h$. By Pythagoras' Theorem,
\begin{equation*} \label{E10}
\Delta{x}^2 (ab') + \Delta{y}^2 (ab') = |ab'|^2.
\end{equation*}
Since the shortest distance between two points is a straight line,
$ \len(\gamma_2) \geq |ab'|$. It follows that
\begin{equation} \label{E9}
L^2 = \len^2(\gamma) \geq \len^2(\gamma_1) = \len^2(\gamma_2) 
\geq |ab'|^2 = (2w -|ab|)^2 + 4h^2. 
\end{equation}
Obviously, $|ab| \leq w$, thus from~\eqref{E9} we deduce that
\begin{equation} \label{E4}
w^2 +4h^2 \leq L^2.
\end{equation}
Substituting $w=L/\lambda$ in~\eqref{E4} yields 
$L^2/\lambda^2 +4h^2 \leq L^2$, hence
\begin{equation*} \label{E5}
2h \leq \frac{\sqrt{\lambda^2-1}}{\lambda} \, L. 
\end{equation*}
It follows that 
\begin{equation} \label{E6}
\per(Q)= 2w + 2h \leq 
\left( \frac{2}{\lambda} + \frac{\sqrt{\lambda^2-1}}{\lambda} \right) L
= \left( \frac{2+\sqrt{\lambda^2-1}}{\lambda} \right) L.
\end{equation}

Consider the function 
\begin{equation*} \label{E7}
f(\lambda) = \frac{2+\sqrt{\lambda^2-1}}{\lambda}.
\end{equation*}
Its derivative is 
\begin{equation*} \label{E8}
f'(\lambda) = \frac
{\frac{\lambda^2}{\sqrt{\lambda^2-1}}-2 - \sqrt{\lambda^2-1}}
{\lambda^2} =\frac
{\lambda^2 -2\sqrt{\lambda^2-1} -(\lambda^2-1)}
{\lambda^2 \sqrt{\lambda^2-1}} =\frac
{1-2\sqrt{\lambda^2-1}}
{\lambda^2 \sqrt{\lambda^2-1}}.
\end{equation*}

It can be easily verified that $f'(\lambda)$ vanishes at
$\lambda=\sqrt5/2$, and that $f(\lambda)$ is increasing 
on the interval $(1,\sqrt5/2]$ and decreasing on the 
interval $[\sqrt5/2,\infty)$. Hence $f(\lambda)$ attains its maximum at 
$\lambda=\sqrt5/2$, that is, 
$f(\lambda) \leq f(\sqrt{5}/2)= \sqrt5$. 
According to~\eqref{E6}, we have $\per(Q) \leq L \sqrt5$, as desired.

To see that this inequality is tight, let $\gamma$ be a two-segment
polygonal path made from the two equal sides of an isosceles triangle
with sides $1$, $1$, and $4/\sqrt5$. In this case, 
$L=2$ and $\per(Q)= 2(4/\sqrt5 + 1/\sqrt5)= 2 \sqrt5 =L \sqrt5$.
It can be verified that the perimeter of any other enclosing rectangle
is larger (details in the Appendix), hence the rectangle $Q$
enclosing $\gamma$ constructed in our proof has minimum perimeter.  
The proof of Lemma~\ref{L5} is now complete.
\end{proof}

Given $\R$, we compute an approximation of the optimal TSP path 
by using algorithm A3. 
Let now $\gamma$ be an optimal TSP path for the rays in $\R$,
where $L=\len(\gamma)$. 
Since every ray in $\R$ intersects $\gamma$, every ray in $\R$
intersects  the rectangle $Q=Q(\gamma)$, as defined in the proof of
Lemma~\ref{L5}. For each of the $m$ directions, the algorithm A3
computes a minimum-perimeter rectangle (of that orientation)
intersecting each ray in $\R$.  Thus A3 computes a rectangle (\ie, a
closed path) that is a $(1+\eps)$-approximation of the minimum-perimeter
rectangle intersecting each ray in $\R$. For $\eps=1/1000$, its
perimeter is at most $(1+\eps) \sqrt5 L \leq 2.24 \, L$, as claimed. 
This completes the proof of Theorem~\ref{T2}.

\section{Final remarks}  \label{sec:remarks}

Interesting questions remain open regarding the structure of 
optimal TSP tours for lines and rays, and the degree of approximation
achievable for these problems. For instance, the following two
problems (one new and one old) deserve attention. 

\begin{itemize}

\item [(1)] Is there a polynomial-time algorithm for computing
a shortest tour (or path) for a given set of rays in the plane?

\item [(2)] Can one compute a good approximation of 
a shortest TSP tour (or path) for a given set of lines in $3$-space?
Note that this variant with parallel lines reduces to the Euclidean TSP
for points in the plane (namely the points of intersection between the
given lines and an orthogonal plane), so it is NP-hard. See also~\cite{DM03}.  

\end{itemize}

\section*{Appendix}

\paragraph{A tight bound for Lemma~\ref{L3}.}
Let $Q$ be a minimal rectangle containing $\gamma= acb$, whose 
width and height are $w$ and $h$, respectively. 
Let $\Delta=\conv(\gamma)= \conv(acb)$, where $\angle{acb}=\pi/2$.
It suffices to show that $w +2h \geq L \sqrt{2} L = 2\sqrt2$. 
By the minimality of $Q$, at least one vertex of $\Delta$ 
must coincide with a corner of $Q$, say the lower left corner $q_1$. 
If $c=q_1$, then $Q$ is a unit square, thus 
$w +2h =3 > 2\sqrt2$. Assume now that $a=q_1$, as in Figure~\ref{f4}. 
\begin{figure}[htbp]
\centerline{\epsfxsize=1.5in \epsffile{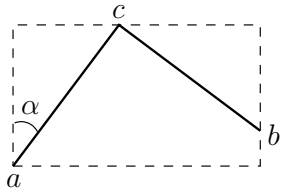}}
\caption{A rectangle containing $\gamma$.}
\label{f4}
\end{figure}

We have $w= \sin \alpha + \cos \alpha$ and $h = \cos \alpha$, where
$\alpha \in [0,\pi/4]$. 
Hence 
$ w +2h = (\sin \alpha + \cos \alpha) + 2 \cos \alpha =
3 \cos \alpha + \sin \alpha$. 
Consider the function $f(\alpha)= 3 \cos \alpha + \sin \alpha$, 
where $\alpha \in [0,\pi/4]$. 
Its derivative, $f'(\alpha)= -3  \sin \alpha + \cos \alpha$ vanishes
at $\alpha = \arctan (1/3)$, and is positive on $(0,\arctan (1/3))$ and
negative on $(\arctan (1/3), \pi/4)$. 
Hence $f(\alpha)$ attains its minimum at one of the 
endpoints of the interval $[0,\pi/4]$. 
We have $f(0)= 3$ and $f(\pi/4)= 4 \sqrt2/2= 2 \sqrt2$,
therefore $f(\alpha) \geq 2 \sqrt2$ for $\alpha \in [0,\pi/4]$.

\paragraph{A tight bound for Lemma~\ref{L5}.}
Again, let $Q$ be a minimal rectangle containing $\gamma= acb$, whose 
width and height are $w$ and $h$, respectively. 
Let $\Delta=\conv(\gamma)= \conv(acb)$, where 
$\angle{acb}=2 \arctan 2$.
It suffices to show that $2(w+h) \geq L \sqrt{5} L = 2\sqrt5$. 
By the minimality of $Q$, at least one of the two vertices $a$ and $b$
of $\Delta$ must coincide with a corner of $Q$, say the lower left
corner $q_1$.  We have $a=q_1$, as in Figure~\ref{f5}. 
We distinguish two cases:
\begin{figure}[htbp]
\centerline{\epsfxsize=4.7in \epsffile{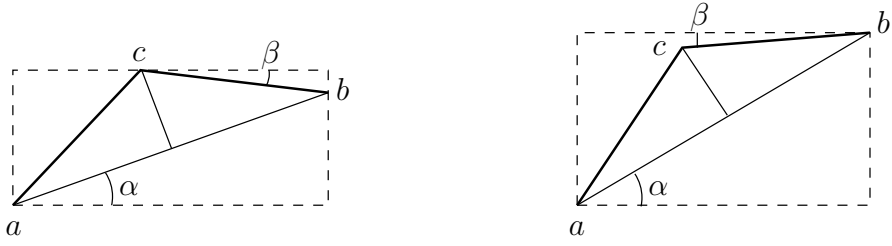}}
\caption{Two type of rectangles containing $\gamma$: Case 1 and Case 2.}
\label{f5}
\end{figure}

\medskip
\emph{Case 1.} The vertex $b$ lies on the right side of $Q$, 
and $c$ lies on the top side of $Q$, as in Figure~\ref{f5}(left). 
We have $w= \frac{4}{\sqrt5} \cos \alpha$ and 
$h = \frac{4}{\sqrt5} \sin \alpha + \sin \beta$, 
where $\alpha \in [0, \arctan (1/2)]$,
and $\beta = \pi/2 - (\pi -(\pi/2 - \alpha) - \arctan (1/2))=
\arctan (1/2) - \alpha$. This yields
\begin{align*}
\sin \beta &= \sin (\arctan (1/2) - \alpha)=
\sin \arctan (1/2) \cos \alpha - \cos \arctan (1/2) \sin \alpha \\
&= \frac{1}{\sqrt5} \cos \alpha - \frac{2}{\sqrt5} \sin \alpha. 
\end{align*}
It follows that 
$$ 2(w+h)= 2 \left(\frac{4}{\sqrt5} \cos \alpha +
\frac{4}{\sqrt5} \sin \alpha + \frac{1}{\sqrt5} \cos \alpha -
\frac{2}{\sqrt5} \sin \alpha \right) =
2 \left( \sqrt5 \cos \alpha + \frac{2}{\sqrt5} \sin \alpha \right). $$
Consider the function 
$f(\alpha)= \sqrt5 \cos \alpha + (2/\sqrt5) \sin \alpha$, 
where $\alpha \in [0,\arctan (1/2)]$. 
Its derivative, 
$f'(\alpha)= -\sqrt5  \sin \alpha + (2/\sqrt5) \cos \alpha$ vanishes
at $\alpha = \arctan (2/5)$, and is positive on $(0,\arctan (2/5))$ and
negative on $(\arctan (2/5), \arctan (1/2))$. 
Hence $f(\alpha)$ attains its minimum at one of the 
endpoints of the interval $[0,\arctan (1/2)]$. 
We have $f(0)= \sqrt5$ and $f(\arctan (1/2))= 2 + 2/5 = 12/5 > \sqrt5$,
therefore $f(\alpha) \geq \sqrt5$ for $\alpha \in [0,\arctan (1/2)]$. 
Equivalently, $\per(Q)= 2(w+h) \geq 2\sqrt5$ in this case.
(It may be noted that the rectangle with perimeter $2 \cdot
12/5=24/5$ corresponding to $\alpha = \arctan (1/2)$ is flush with the
unit side $cb$.) 

\medskip
\emph{Case 2.} 
The vertex $b$ coincides with the upper right corner of $Q$, 
and $c$ lies in the interior of $Q$, as in Figure~\ref{f5}(right). 
By symmetry, we can assume that $w \geq h$.
We have $w= (4/\sqrt5) \cos \alpha)$ and $h = (4/\sqrt5) \sin \alpha$, where
$\alpha= \beta + \arctan (1/2)$, and $\beta \in [0, \pi/4 - \arctan (1/2)]$. 
This yields 
$$ w+h= \frac{4}{\sqrt5} 
(\cos (\beta + \arctan (1/2)) + \sin (\beta + \arctan (1/2))). $$
The above expression attains its minimum at $\beta=0$, thus
\begin{align*}
\per(Q) &= 2(w+h) \geq 
\frac{8}{\sqrt5} \left( \cos \arctan (1/2) + \sin \arctan (1/2) \right) \\
&= \frac{8}{\sqrt5} \left( \frac{2}{\sqrt5} + \frac{1}{\sqrt5} \right) 
= \frac{8}{\sqrt5} \cdot \frac{3}{\sqrt5} = \frac{24}{5} > 
2\sqrt5, 
\end{align*}
in this second case.

\medskip
We therefore always have $\per(Q) \geq 2 \sqrt5$, as claimed.

\end{document}